\documentclass[3p,review]{elsarticle}

\usepackage{lineno,hyperref,graphicx,epstopdf}
\usepackage{times}
\usepackage{type1cm}
\fontsize{12pt}{18pt}

\usepackage{braket}
\usepackage{amsmath}
\usepackage{amsfonts}
\usepackage{amssymb}
\usepackage{amsthm}
\usepackage{algorithm}
\usepackage{algorithmic}

\theoremstyle{plain}
\newtheorem{theorem}{Theorem}

\newtheorem{lemma}{Lemma}
\newtheorem{proposition}{Proposition}

\theoremstyle{definition}
\newtheorem{definition}{Definition}

\newtheorem{remark}{Remark}

\journal{}

\date{}

\begin{document}
\begin{frontmatter}
\title{Distributed Shor's algorithm}

\author{Ligang Xiao$^{1}$}
\author{Daowen Qiu$^{1,3,}$\corref{one}}
\author{Le Luo$^{2,3}$}
\author{Paulo Mateus$^{4}$}
\cortext[one]{Corresponding author (D. Qiu). {\it E-mail addresses:} issqdw@mail.sysu.edu.cn (D. Qiu)}

\address{
   %$^1$
   $^{1}$Institute of Quantum Computing and Computer Theory, School of Computer Science and Engineering, Sun Yat-sen University, Guangzhou 510006, China\\
   $^{2}$School of Physics and Astronomy, Sun Yat-sen University, Zhuhai 519082, China\\
 $^{3}$QUDOOR Technologies Inc., Guangzhou,  China \\
	 $^{4}$Instituto de Telecomunica\c{c}\~{o}es, Departamento de Matem\'{a}tica,
	Instituto Superior T\'{e}cnico,  Av. Rovisco Pais 1049-001  Lisbon, Portugal}

\begin{abstract}
Shor's algorithm is one of the most important quantum algorithm proposed by Peter Shor [Proceedings of the 35th Annual Symposium on Foundations of Computer Science, 1994, pp. 124--134]. Shor's algorithm can factor a large integer with certain probability and costs  polynomial time in the length of the input integer. The key step of Shor's algorithm is the order-finding algorithm. Specifically, given an $L$-bit integer $N$, we first randomly pick an integer $a$ with $gcd(a,N)=1$, the order of $a$ modulo $N$ is the smallest positive integer $r$ such that $a^r\equiv 1 (\bmod N)$. The order-finding algorithm in Shor's algorithm first uses quantum operations to obtain an estimation of $\dfrac{s}{r}$ for some $s\in\{0, 1, \cdots, r-1\}$, then  $r$ is obtained by means of classical algorithms. In this paper, we propose a distributed Shor's algorithm. The difference between our distributed algorithm and the traditional order-finding algorithm is that we use two quantum computers separately to estimate partial bits of $\dfrac{s}{r}$ for some $s\in\{0, 1, \cdots, r-1\}$. To ensure their measuring results correspond to the same $\dfrac{s}{r}$, we need employ quantum teleportation. We integrate the measuring results via classical post-processing. After that, we get an estimation of $\dfrac{s}{r}$ with high precision. Compared with the traditional Shor's algorithm that uses multiple controlling qubits, our algorithm reduces nearly $\dfrac{L}{2}$ qubits and reduces the circuit depth of each computer.\\
\end{abstract}

\begin{keyword}
%% keywords here, in the form: keyword \sep keyword
Shor's algorithm \sep Distributed Shor's algorithm \sep Quantum teleportation \sep Circuit depth \
\end{keyword}

\end{frontmatter}

%取消注释这行可显示行号
%\linenumbers

\section{Introduction}\label{sec:introduction}
	Quantum computing has shown great potential in some fields or problems, such as chemical molecular simulation \cite{aspuru2005simulated}, portfolio optimization \cite{rosenberg2016solving}, large number decomposition \cite{shor1994algorithms}, unordered database search \cite{grover1996a} and linear equation solving \cite{harrow2009quantum} et al. At present, there have been many useful algorithms in quantum computing \cite{montanaro2016quantum}, but to realize these algorithms requires the power of large-scale general quantum computers. However, it is still very difficult to develop a large-scale general quantum computer, because there are important physical problems in quantum computer that have not been solved. Therefore it is necessary to consider reducing the number of qubits and other computing resources required for quantum algorithms.

	Distributed quantum computing is a computing method that solves problems collaboratively through multiple computing nodes. In distributed quantum computing, we can use multiple medium-scale quantum computers to complete a task that was originally completed by a single large-scale quantum computer. Distributed quantum computing not only reduces the number of qubits required, but also sometimes reduces the circuit depth of each computer. This is also important since noise is increased with circuit being deepened . Therefore, distributed quantum computing has been studied significantly (for example, \cite{avron2021quantum,beals2013effcient,li2017application,yimsiriwattana2004distributed}).

	Shor's algorithm proposed by Peter Shor in 1994 \cite{shor1994algorithms} is an epoch-making discovery. It can factor a large integer with certain probability and costs time polynomial in the length of the input integer, whereas the time complexity of the best known classical algorithm for factoring large numbers is exponential. Shor's algorithm can be applied in cracking various cryptosystems, such as RSA cryptography and elliptic curve cryptography. For this reason, Shor's algorithm has received extensive attention from the community. However, recently some researchers have pointed out that using Shor's algorithm to crack the commonly used 2048-bit RSA integer requires  physical qubits  of millions \cite{craig2021how}. So it is vital to consider reducing the qubits required in Shor's algorithm. Many researchers have been working on reducing the number of qubits required for Shor's algorithm \cite{beauregard2003circuit, haner2017factoring,parker2000efficient}, and these results have shown that Shor's algorithm can be implemented using only one controlling qubit to factor a $L$-bit integer together with $2L+c$ qubits and  circuit depth $O(L^3)$, where $c$ is a constant.

	In 2004, Yimsiriwattana et al \cite{yimsiriwattana2004distributed} proposed a distributed Shor's algorithm. In this distributed algorithm, it directly divides the qubits into several parts reasonably, so each part has fewer qubits than the original one. Since all unitary operators can be decomposed into single qubit quantum gates and CNOT gates \cite{nielsen2000quantum}, they only need to consider how to implement CNOT gates acting on different parts, while a CNOT gate acting on different parts can be implemented by means of pre-sharing EPR pairs, local operations and classical communication. Their distributed algorithm needs to communicate $O(L^2)$ classical bits.

	In this paper, we propose a new distributed Shor's algorithm. In our distributed algorithm, two computers execute sequentially. Each computer estimates several bits of some key intermediate quantity. In order to guarantee the correlation between the two computers' measuring results to some extent, we employ quantum communication. Furthermore, to obtain high accuracy, we can adjust the measuring result of the first computer in terms of the measuring result of the second computer through classical post-processing. Compared with the traditional Shor's algorithm that uses multiple controlling qubits, our algorithm reduces the cost of qubits (reduces nearly $\dfrac{L}{2}$ qubits) and the circuit depth of each computer. Although each computer in our distributed algorithm requires more qubits than the Shor's algorithm mentioned above that uses only one controlling qubit, our method of using quantum communication to distribute the phase estimation  of Shor's algorithm may be applicable to other quantum algorithms.

	The remainder of the paper is organized as follows. In Section \ref{sec:preliminaries}, we review quantum teleportation and some quantum algorithms related to Shor's algorithm. In Section \ref{sec:distributed order finding algorithm}, we present a distributed Shor's algorithm (more specifically, a distributed order-finding algorithm), and prove the correctness of our algorithm. In Section \ref{sec:complexity analysis}, we analyze the performance of our algorithm, including space complexity, time complexity, circuit depth and communication complexity. Finally in Section \ref{sec:conclusions}, we conclude with a summary.

\section{Preliminaries}\label{sec:preliminaries}

	In this section, we review the quantum Fourier transform, phase estimation algorithm, order-finding algorithm and others we will use. We assume that the readers are familiar with the liner algebra and basic notations in quantum computing (for the details we can refer to \cite{nielsen2000quantum}).

\subsection{Quantum Fourier transform}

	Quantum Fourier transform is a unitary operator with the following action on the standard basis states:
\begin{equation}
QFT |j\rangle=\frac{1}{\sqrt{2^n}}\sum_{k=0}^{2^n-1}e^{2\pi ijk/2^n}|k\rangle\text{,}
\end{equation}
for $j=0,1,\cdots,2^n-1$. Hence the inverse quantum Fourier transform is acted  as follows:
\begin{equation}\label{inverse_QFT}
QFT^{-1} \frac{1}{\sqrt{2^n}}\sum_{k=0}^{2^n-1}e^{2\pi ijk/2^n}|k\rangle=|j\rangle\text{,}
\end{equation}
for $j=0,1,\cdots,2^n-1$.
Quantum Fourier transform and the inverse quantum Fourier transform can be implemented using $O(n^2)$ single elementary gates (i.e. $O(n^2)$ single qubit gates and two-qubit gates) \cite{nielsen2000quantum, shor1994algorithms}.

\subsection{Phase estimation algorithm}

	Phase estimation algorithm is an application of the quantum Fourier transform. Let $|u\rangle$ be a quantum state and let $U$ be a unitary operator that satisfies $U|u\rangle=e^{2\pi i\omega}$ for some real number $\omega\in[0,1)$. Suppose we can create the quantum state $|u\rangle$ and implement controlled operation $C_m(U)$ such that
\begin{equation}
C_m(U)|j \rangle|u\rangle=|j\rangle U^j|u\rangle
\end{equation}
for any positive integer $m$ and $m$-bit string $j$, where the first register is control qubits. Figure \ref{fig:controlled-U} shows the implementation of $C_m(U)$. Then we can apply phase estimation algorithm to estimate $\omega$. For the sake of convenience, we first define the following notations. In this paper, we treat bit strings and their corresponding binary integers as the same.

\begin{figure}[h!]
	\centering
	\includegraphics[width=0.5\textwidth]{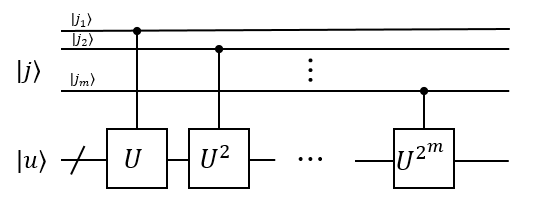}
	\caption{\label{fig:controlled-U} Implementation for $C_m(U)$}
\end{figure}

\begin{definition}
For any real number $\omega=a_{1}a_{2}\cdots a_l.b_1b_2\cdots$, where $a_{k_1}\in\{0,1\}, k_1=1,2,\cdots, l$ and $b_{k_2}\in\{0,1\}, k_2=1,2,\cdots$, denote $|\psi_{t,\omega}\rangle, FBits(\omega,i,j), IBits(\omega,i,j)$ and $d_t(x,y)$ respectively as follows:
\begin{itemize}
\item $|\psi_{t,\omega}\rangle$: for any positive integer $t$, $|\psi_{t,\omega}\rangle=QFT^{-1}\dfrac{1}{\sqrt{2^t}}\sum\limits_{j=0}^{2^t-1}e^{2\pi ij\omega}|j\rangle $ .
\item $FBits(\omega,i,j)$: for any integer $i,j$ with $1\leq i\leq j$, $FBits(\omega,i,j)=b_i b_{i+1}\cdots b_j$.
\item $IBits(\omega,i,j)$: for any integer $i,j$ with $1\leq i\leq j\leq l$, $IBits(\omega,i,j)=a_{i} a_{i+1}\cdots a_j$.
\item $d_t(x,y)$: for any two $t$-bit strings (or $t$-bit binary integers) $x,y$, define $d_t(x,y)=\min(|x-y|, 2^t-|x-y|)$.
\end{itemize}
\end{definition}

$d_t(\cdot, \cdot)$ is a useful distance to estimate the error of the algorithms in our paper and it has the following properties.  We specify $a \bmod N=(kN+a) \bmod N$ for any negative integer $a$ and positive integer $N$, where $k$ is an integer and satisfies $kN+a\geq 0$.

\begin{lemma}\label{d_t}
Let $t$ be a positive integer and let $x,y$ be any two $t$-bit strings. It holds that: \\
{\rm (I)} Let $B=\{b\in\{-(2^t-1),\cdots,2^{t}-1\}: (x+b)\bmod 2^t=y\}$. Then $d_t(x,y)=\min_{b\in B}|b|$.\\
{\rm (II)} $d_t(\cdot,\cdot)$ is a distance on $\{0,1\}^t$.\\
{\rm (III)} Let $t_0<t$ be an positive integer. If $d_t(x,y)<2^{t-t_0}$, then
\begin{equation}\label{d_t 1}
d_{t_0}(IBits(x,1,t_0),IBits(y,1,t_0))\leq 1.
\end{equation}
\end{lemma}
\begin{proof}
First we prove (I). It is clear for the case of $x=y$. Without loss of generality, assume $x>y$. Since $x\not=y$, we have $B$ contains only $2$ elements. Note that

\begin{align}
x+(y-x)\bmod 2^t=y,\\
x+(2^t-(x-y))\bmod 2^t=y,\\
 |y-x|\leq 2^t-1,\\
 |2^t-(x-y)|\leq 2^t-1
\end{align}
and $y-x\not=2^t-(x-y)$,
we get that $y-x$ and $2^t-(x-y)$ are exactly two elements of $B$. Hence $\min_{b\in B}|b|=\min(|x-y|, 2^t-|x-y|)=d_t(x,y)$. Thus {\rm (I)} holds.

Then we prove $\rm (II)$. We just need to show that $d_t(\cdot,\cdot)$ satisfies the triangle inequality, that is, $d_t(x,y)\leq d_t(x,z)+d_t(z,y)$ holds for any $t$-bit string $z$. By {\rm (I)}, we know that there exists $b_1,b_2\in \{-(2^t-1),\cdots,2^{t}-1\}$ such that
\begin{equation}
|b_1|=d_t(x,z),|b_2|=d_t(z,y),
\end{equation}
and
\begin{equation}
(x+b_1)\bmod 2^t=z,(z+b_2)\bmod 2^t=y.
\end{equation}
Hence $(x+b_1+b_2)\bmod 2^t=y$.
Then by (I) again, we have
\begin{equation}
d_t(x,y)\leq |b_1+b_2|\leq |b_1|+|b_2|= d_t(x,z)+d_t(z,y).
\end{equation}
Thus, (II) holds.

Finally we prove $\rm (III)$. By (I) and $d_t(x,y)<2^{t-t_0}$, we know that there exists an integer $b$ with $|b|<2^{t-t_0}$ such that
\begin{equation}
((2^{t-t0}IBits(x,1,t_0)+IBits(x,t_0+1,t)+b)\bmod 2^t=2^{t-t0}IBits(y,1,t_0)+IBits(y,t_0+1,t).
\end{equation}
Then by (I) again we have
\begin{equation}
d_t(2^{t-t0}IBits(x,1,t_0),2^{t-t0}IBits(y,1,t_0))\leq |b+IBits(x,t_0+1,t)-IBits(y,t_0+1,t)|<2\cdot 2^{t-t0}.
\end{equation}
Hence
\begin{equation}
d_{t_0}(IBits(x,1,t_0),IBits(y,1,t_0))< 2.
\end{equation}
Therefore Equation (\ref{d_t 1}) holds.
\end{proof}

We can understand $d_t(\cdot,\cdot)$ in a more intuitive way. We place numbers 0 to $2^t$ evenly on a circumference where $0$ and $2^t$ coincide. Suppose the distance of two adjacent points on the circumference is $1$. Then $d_t(x,y)$ can be regarded as the length of the shortest path on the circumference from $x$ to $y$. Next we review the phase estimation algorithm (see Algorithm \ref{Phase_estimation_algorithm}) and its associated results.

\begin{algorithm}[h]
\caption{Phase estimation algorithm}
\label{Phase_estimation_algorithm}
%\textbf{Input}: The probability of success $1-\epsilon$ and initial state $|0\rangle^{\otimes t}|u\rangle$, where $t=n+\lceil\log_2(2+\dfrac{1}{2\epsilon})\rceil$.\\
%\textbf{Output}: An estimation $\widetilde{\omega}$ of $\omega$ with $|\widetilde{\omega}-\omega|<\dfrac{1}{2^n}$.\\
\textbf{Procedure}:
\begin{algorithmic}[1]
\STATE Create initialize state $|0\rangle^{\otimes t}|u\rangle$.
\STATE Apply $H^{\otimes t}$ to the first register:\\
  \quad$H^{\otimes t}|0\rangle^{\otimes t}|u\rangle=\dfrac{1}{\sqrt{2^t}}\sum\limits_{j=0}^{2^t-1}|j\rangle|u\rangle$.
\STATE Apply $C_t(U)$:\\
  \quad$C_t(U)\dfrac{1}{\sqrt{2^t}}\sum\limits_{j=0}^{2^t-1}|j\rangle|u\rangle=\dfrac{1}{\sqrt{2^t}}\sum\limits_{j=0}^{2^t-1}|j\rangle U^j|u\rangle=\dfrac{1}{\sqrt{2^t}}\sum\limits_{j=0}^{2^t-1}|j\rangle e^{2\pi ij\omega}|u\rangle$.
\STATE Apply $QFT^{-1}$:\\
 \quad$QFT^{-1}\dfrac{1}{\sqrt{2^t}}\sum\limits_{j=0}^{2^t-1}e^{2\pi ij\omega}|j\rangle |u\rangle=|\psi_{t,\omega}\rangle|u\rangle$.
\STATE Measure the first register:\\
 \quad obtain a $t$-bit string $\widetilde{\omega}$.

\end{algorithmic}
\end{algorithm}

	If the fractional part of $\omega$ does not exceed $t$ bits (i.e. $2^t\omega$ is an integer), by observing Equation (\ref{inverse_QFT}) and the step 4 in Algorithm \ref{Phase_estimation_algorithm}, we can see that $\widetilde{\omega}$ is a perfect estimate of $\omega$ (i.e. $\dfrac{\widetilde{\omega}}{2^t}=\omega$). However, sometimes $\omega$ is not approximated by $\dfrac{\widetilde{\omega}}{2^t}$ but is approximated by $1-\dfrac{\widetilde{\omega}}{2^t}$. For example, if the binary representation of $\omega$ is $\omega=0.11\cdots1$ (sufficiently many $1$s), we will obtain the measuring result $00\cdots 0$ with high probability.  The purpose of  phase estimation algorithm is to find a $\widetilde{\omega}$ such that $\dfrac{\widetilde{\omega}}{2^t}$ is close to $\omega$ or $\omega-1$. We have the following results.

\begin{proposition}[See \cite{nielsen2000quantum}]\label{phase_estimation_result}
	In {\rm Algorithm \ref{Phase_estimation_algorithm}}, for any $\epsilon>0$ and any positive integer $n$, if $t=n+\lceil\log_2(2+\dfrac{1}{2\epsilon})\rceil$, then the probability of $d_t(\widetilde{\omega},FBits(\omega,1,t))<2^{t-n}$ is at least $1-\epsilon$.
\end{proposition}

\begin{lemma}\label{result_acuracy}
For any $t$-bit string $\widetilde{\omega}$ and real number $\omega\in [0,1)$. If $d_t(\widetilde{\omega},FBits(\omega,1,t))<2^{t-n}$, then we have $|\dfrac{\widetilde{\omega}}{2^t}-\omega|\leq 2^{-n}$ or $1-|\dfrac{\widetilde{\omega}}{2^t}-\omega|\leq 2^{-n}$, where $n< t$.
\end{lemma}
\begin{proof}
Since $|2^t\omega-FBits(\omega,1,t)|<1$, if $d_t(\widetilde{\omega},FBits(\omega,1,t))=|\widetilde{\omega}-FBits(\omega,1,t)|$, we have
\begin{equation}
|\widetilde{\omega}-2^t\omega|\leq |\widetilde{\omega}-FBits(\omega,1,t)|+|FBits(\omega,1,t)-2^t\omega|\leq2^{t-n},
\end{equation}
and thus $|\dfrac{\widetilde{\omega}}{2^t}-\omega|\leq 2^{-n}$;
if $d_t(\widetilde{\omega},FBits(\omega,1,t))=2^t-|\widetilde{\omega}-FBits(\omega,1,t)|$, we have
\begin{equation}
2^t-|\widetilde{\omega}-2^t\omega|\leq 2^t-(|\widetilde{\omega}-FBits(\omega,1,t)|-|FBits(\omega,1,t)-2^t\omega|)\leq2^{t-n},
\end{equation}
and therefore, we have $1-|\dfrac{\widetilde{\omega}}{2^t}-\omega|\leq 2^{-n}$.
\end{proof}

That is to say, if $\dfrac{\widetilde{\omega}}{2^t}$ is an estimate of $FBits(\omega,1,t)$ with error less than $2^{-n}$, then $\dfrac{\widetilde{\omega}}{2^t}$ is an estimate of $\omega$ with error no larger than $2^{-n}$.

\subsection{Order-finding algorithm}

 Phase estimation algorithm is a key subroutine in order-finding algorithm. Given an $L$-bit integer $N$ and a positive integer $a$ with $gcd(a,N)=1$, the purpose of order-finding algorithm is to find the order $r$ of $a$ modulo $N$, that is, the least integer $r$ that satisfies $a^r\equiv 1(\bmod\ N)$. An important unitary operator  $M_a$ in order-finding algorithm is defined as
\begin{equation}
M_a|x \rangle=|ax\ \bmod\ N\rangle \text{.}
\end{equation}
Denote
\begin{equation}
|u_s\rangle=\dfrac{1}{\sqrt{r}}\sum\limits_{k=0}^{r-1}e^{-2\pi i\frac{s}{r}k}|a^k \bmod\ N\rangle\text{.}
\end{equation}
We have
\begin{align}
M_a|u_s\rangle=e^{2\pi i\frac{s}{r}} |u_s\rangle,\\
\dfrac{1}{\sqrt{r}}\sum\limits_{s=0}^{r-1}|u_s\rangle=|1\rangle,
\end{align} and
\begin{equation}\label{us_orthonormal}
\langle u_s|u_{s'}\rangle=\delta_{s,s'}=
\begin{cases} 0 &\text{if $s\not= s'$},\\
1 &\text{if $s=s'$}.
\end{cases}
\end{equation}
So if we expect to apply phase estimation algorithm in finding order, the key is to construct $C_m(M_a)$, that is, for any $m$-bit string $j$,
\begin{equation}
C_m(M_a)|j\rangle|x \rangle=|j\rangle|a^jx \bmod\ N\rangle\text{.}
\end{equation}
Algorithm \ref{Order_finding_algorithm}  \cite{nielsen2000quantum} and Figure \ref{fig:order finding algorithm} show the precedure of  order-finding algorithm.

\begin{algorithm}[h!]
\caption{Order-finding algorithm}
\label{Order_finding_algorithm}
\textbf{Input}: Positive integers $N$ and $a$ with $gcd(N,a)=1$.\\
\textbf{Output}: The order $r$ of $a$ modulo $N$.\\
\textbf{Procedure}:
\begin{algorithmic}[1]
\STATE Create initial state $|0\rangle^{\otimes t}|1\rangle$:\\
\quad $t=2L+1+\lceil\log_2(2+\dfrac{1}{2\epsilon})\rceil$ and the second register has $L$ qubits.
\STATE Apply $H^{\otimes t}$ to the first register:\\
 \quad $H^{\otimes t}|0\rangle^{\otimes t}|1\rangle=\dfrac{1}{\sqrt{2^t}}\sum\limits_{j=0}^{2^t-1}|j\rangle|1\rangle$.
\STATE Apply $C_t(M_a)$:\\
  \quad$C_t(M_a)\dfrac{1}{\sqrt{2^t}}\sum\limits_{j=0}^{2^t-1}|j\rangle|1\rangle=\dfrac{1}{\sqrt{2^t}}\sum\limits_{j=0}^{2^t-1}|j\rangle M^j (\dfrac{1}{\sqrt{r}}\sum\limits_{s=0}^{r-1}|u_s\rangle)=\dfrac{1}{\sqrt{r2^t}}\sum\limits_{s=0}^{r-1}\sum\limits_{j=0}^{2^t-1}|j\rangle e^{2\pi ij\frac{s}{r}}|u_s\rangle$.
\STATE Apply $QFT^{-1}$:\\
 \quad $QFT^{-1}\dfrac{1}{\sqrt{r2^t}}\sum\limits_{s=0}^{r-1}\sum\limits_{j=0}^{2^t-1}|j\rangle e^{2\pi ij\frac{s}{r}}|u_s\rangle=\dfrac{1}{\sqrt{r}}\sum\limits_{s=0}^{r-1}|\psi_{t,s/r}\rangle|u_s\rangle$
\STATE Measure the first register:\\
\quad obtain a $t$-bit string $m$ that is an estimation of $\dfrac{s}{r}$ for some $s$.
\STATE Apply continued fractions algorithm:\\
\quad obtain $r$.
\end{algorithmic}
\end{algorithm}

\begin{figure}[h!]
	\centering
	\includegraphics[width=0.6\textwidth]{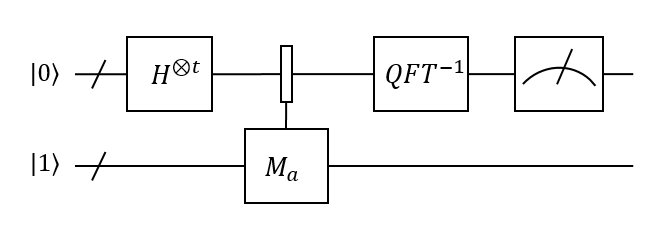}
	\caption{\label{fig:order finding algorithm} Circuit for order-finding algorithm}
\end{figure}

	The purpose of steps 1 to 5 in Algorithm \ref{Order_finding_algorithm} is to get a measuring result $m$ such that $m$ is an estimation of $\dfrac{s}{r}$ for some $s\in\{0,1,\cdots,r-1\}$ (i.e. $|\dfrac{m}{2^t}-\dfrac{s}{r}|\leq 2^{-(2L+1)}$). Let $\{P_i\}$ be any projective measurement on $\mathbb{C}^{2^t}$ and let $|\phi_s\rangle$ be any $t$-qubit quantum state for $s=0,1,\cdots,r-1$. By Equation (\ref{us_orthonormal}), we have
\begin{equation}\label{measure_orthonormal}
\|(P_j\otimes I)\sum_{s=0}^{r-1} |\phi_s\rangle|u_s\rangle\|^2=\sum_{s=0}^{r-1}\|(P_j |\phi_s\rangle)|u_s\rangle\|^2
\end{equation}
for $P_j\in \{P_i\}$. Hence by Propositon \ref{phase_estimation_result} and Equation (\ref{measure_orthonormal}), we can obtain the following proposition immediately.

\begin{proposition}[See \cite{nielsen2000quantum}]\label{order_finding_result}
In {\rm Algorithm \ref{Order_finding_algorithm}}, the probability of
 $d_t(m,Bits(\dfrac{s}{r},t))<2^{t-(2L+1)}$ for any fixed $s\in \{0,1,\cdots,r-1\}$ is at least $\dfrac{1-\epsilon}{r}$. Thus the probability that there exists an $s\in \{0,1,\cdots,r-1\}$ such that
\begin{equation}
d_t(m,FBits(\dfrac{s}{r},1,t))<2^{t-(2L+1)}
\end{equation} is at least $1-\epsilon$.
\end{proposition}

	Although it is an important part to discuss the probability of obtaining $r$ correctly from the measuring result by applying continued fractions algorithm, the details are omitted here and we focus on considering whether the measuring result is an estimation of $\dfrac{s}{r}$ for some $s$, since this is exactly the goal of the quantum part in the order-finding algorithm.

\subsection{Quantum teleportation}
	Quantum teleportation is an important means to realize quantum communication \cite{bennett1993teleporting,ni2007quantum}. Quantum teleportation is effectively equivalent to physically teleporting qubits, but in fact, the realization of quantum teleportation only requires classical communication and both parties to share an EPR pair in advance. The following result is useful.
\begin{theorem}[\cite{bennett1993teleporting}]\label{quantum_telepotation}
When Alice and Bob share $L$ pairs of EPR pairs, they can simulate transmitting $L$ qubits by communicating $2L$ classical bits.
\end{theorem}

\section{Distributed order-finding algorithm} \label{sec:distributed order finding algorithm}
	In \cite{li2017application}, a  distributed phase estimation algorithm was proposed, but the method in \cite{li2017application} can not guarantee the precision of the result. However, their ideas deserve further consideration. In this section, by combining with quantum teleportation, we proposed a distributed order-finding algorithm and prove the correctness of our algorithm.

	Without loss of generality, assume that $L=\lceil\log_2(N)\rceil$ is even. The idea of our distributed order-finding algorithm is as follows. We need two quantum computers (named $A$ and $B$). We first apply order-finding algorithm in computer $A$ and obtain an estimation of the first $\dfrac{L}{2}+1$ bits of $\dfrac{s}{r}$ for some $s\in\{0,1,\cdots,r-1\}$, and similarly obtain an estimation of the $(\dfrac{L}{2}+2)$th bit to $(2L+1)$th bit of $\dfrac{s}{r}$ in computer $B$. We can realize this by using $C_t(M_a^{2^l})$, since $M_a^{2^l}|u_s\rangle=e^{2\pi i(2^l\frac{s}{r})} |u_s\rangle$ and the fractional part of $2^l\dfrac{s}{r}$ starts at the $(l+1)$th bit of he fractional part of $\dfrac{s}{r}$. Moreover, since $M_a^{2^l}=M_{a^{2^l} {\rm mod}\ N}$ and we can calculate $a^{2^l} \bmod\ N$ classically with time complexity $O(l)$,  we can construct $C_t(M_a^{2^l})$ with the same way as $C_t(M_a)$. In addition, to guarantee the measuring results of $A$ and $B$ corresponding to the same $\dfrac{s}{r}$, we need quantum teleportation.

	However, in order to maintain high precision, computer $B$ actually estimates the $\dfrac{L}{2}$th bit to  $(2L+1)$th bit, where the estimation of the $\dfrac{L}{2}$th bit and the $(\dfrac{L}{2}+1)$th bit is used to ``correct" the measuring result of $A$. This ``correction" operation is handed over to a classical subroutine named $CorrectResults$. Our distributed order-finding algorithm is shown in Algorithm \ref{Distributed_order_finding} and  Figure \ref{fig:distributed order finding algorithm}, and the subroutine $CorrectResults$ is shown in Algorithm \ref{CorrectResults subroutine}

\begin{figure}[h!]
	\centering
	\includegraphics[width=1\textwidth]{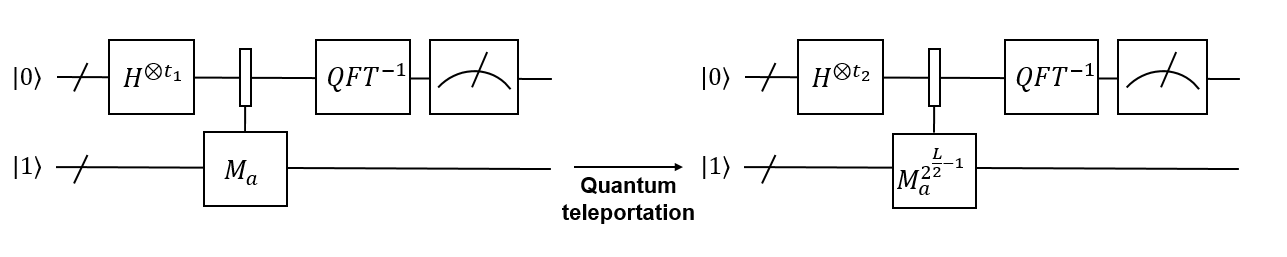}
	\caption{\label{fig:distributed order finding algorithm} Circuit for distributed order finding algorithm}
\end{figure}

\begin{algorithm}[h!]
\caption{Distributed order-finding algorithm}
\label{Distributed_order_finding}
\textbf{Input}: Positive integers $N$ and $a$ with $gcd(N,a)=1$.\\
\textbf{Output}: The order $r$ of $a$ modulo $N$.\\
\textbf{Procedure}:
\begin{algorithmic}[1]
\STATE Computer $A$ creates initial state $|0\rangle_{A}|1\rangle_C$. Computer $B$ creates initial state $|0\rangle_{B}$: \\
 \quad Here registers $A$, $B$ and $C$ are $t_1$-qubit, $t_2$-qubit and $L$-qubit, respectively. We take \\$t_1=\dfrac{L}{2}+1+p$ and $t_2=\dfrac{3L}{2}+2+p$, where $p=\lceil\log_2(2+\dfrac{1}{2\epsilon'}\rceil$ and $\epsilon'=\dfrac{\epsilon}{2}$.\\

\textbf{Computer  $A$}:
\STATE Apply $H^{\otimes t_1}$ to register $A$:
 \quad $\rightarrow \dfrac{1}{\sqrt{r}}\sum\limits_{s=0}^{r-1} (H^{\otimes t_1}|0\rangle|u_s\rangle) |0\rangle$.

\STATE Apply $C_{t_1}(M_a)$ to registers $A$ and $C$:
  \quad$\rightarrow \dfrac{1}{\sqrt{r}}\sum\limits_{s=0}^{r-1} (\dfrac{1}{\sqrt{2^{t_1}}}\sum\limits_{j=0}^{2^{t_1}-1}e^{2\pi ij\frac{s}{r}}|j\rangle |u_s\rangle) |0\rangle$.

\STATE Apply $QFT^{-1}$to register $A$:
 \quad $\rightarrow \dfrac{1}{\sqrt{r}}\sum\limits_{s=0}^{r-1}|\psi_{t_1,s/r}\rangle|u_s\rangle  |0\rangle$

\STATE Teleport the qubits of register $C$ to computer $B$:
 \quad  $\rightarrow \dfrac{1}{\sqrt{r}}\sum\limits_{s=0}^{r-1}|\psi_{t_1,s/r}\rangle  |0\rangle|u_s\rangle$

\textbf{Computer  $B$}:
\STATE Apply $H^{\otimes t_2}$ to register $B$:
\quad $\rightarrow \dfrac{1}{\sqrt{r}}\sum\limits_{s=0}^{r-1}|\psi_{t_1,s/r}\rangle  H^{\otimes t_2}|0\rangle|u_s\rangle$

\STATE Apply $C_{t_2}(M_a^{2^{\frac{L}{2}-1}})$ to registers $B$ and $C$:
\quad $\rightarrow \dfrac{1}{\sqrt{r}}\sum\limits_{s=0}^{r-1}|\psi_{t_1,s/r}\rangle (\dfrac{1}{\sqrt{2^{t_2}}}\sum\limits_{j=0}^{2^{t_2}-1}e^{2\pi ij(2^{\frac{L}{2}-1}\frac{s}{r})}|j\rangle) |u_s\rangle$

\STATE Apply $QFT^{-1}$to register $B$:
 \quad $\rightarrow |\phi_{final}\rangle=\dfrac{1}{\sqrt{r}}\sum\limits_{s=0}^{r-1}|\psi_{t_1,s/r}\rangle |\psi_{t_2,2^{\frac{L}{2}-1}s/r}\rangle |u_s\rangle$
\STATE Computer $A$ measures register $A$ and computer $B$ measures register $B$:\\
\quad $A$ obtains a $t_1$-bit string $m_1$ and $B$ obtains a $t_2$-bit string $m_2$.
\STATE $m\leftarrow CorrectResults(m_1,m_2)$: $m$ is a $(2L+1+p)$-bit string.
\STATE Apply continued fractions algorithm: obtain $r$.
\end{algorithmic}
\end{algorithm}

\begin{algorithm}[h!]
\caption{CorrectResults subroutine}
\label{CorrectResults subroutine}
\textbf{Input}: Two measuring results: $t_1$-bit string $m_1$ and $t_2$-bit string $m_2$.\\
\textbf{Output}: An estimation $m$ such that $|\dfrac{m}{2^{(2L+1+p)}}-\dfrac{s}{r}|\leq 2^{-(2L+1)}$ for some $s\in\{0,1,\cdots,r-1\}$.\\
\textbf{Procedure}:
\begin{algorithmic}[1]
\STATE Choose $CorrectionBit\in\{-1,0,1\}$ such that\\
\quad $(IBits(m_1,\dfrac{L}{2},\dfrac{L}{2}+1)+ CorrectionBit) \bmod 2^2=IBits(m_2,1,2)$.
\STATE $m_{prefix}\leftarrow (IBits(m_1,1,\dfrac{L}{2}+1)+ CorrectionBit)\bmod 2^{\frac{L}{2}+1}$
\STATE $m\leftarrow m_{prefix}\circ IBits(m_2,3,t_2)$ (``$\circ$" represents catenation)
\STATE return $m$
\end{algorithmic}
\end{algorithm}

\begin{remark}
Although Algorithm \ref{Distributed_order_finding} is a serial algorithm, the two computers can also execute in parallel to some extent. For example, execute the algorithm in the following order: 1, (2, 6), 3, 5, 7, (4, 8), 9, 10, 11, where $i$ represents the $i$th step in Algorithm \ref{Distributed_order_finding}, and $(i, j)$ means that the $i$th and $j$th steps are executed in parallel.
\end{remark}

	%It is worth mentioning that sometimes the input of the subroutine $CorrectResults$ will cause the $CorrectionBit$ to not exist.  But no matter what we do with these ``illegal" inputs, it will not have an essential effect on our algorithm.

	Next we prove the correctness of our algorithm, that is, we can obtain the output $m$ such that $|\dfrac{m}{2^{(2L+1+p)}}-\dfrac{s}{r}|\leq 2^{-(2L+1)}$ holds for some $s\in\{0,1,\cdots,r-1\}$ with high probability. Let $r ,L, t_1,t_2,p,m_1,m_2,m_{prefix}, m,\epsilon',|\phi_{final}\rangle$ be the same as those in Algorithm \ref{Distributed_order_finding} and Algorithm \ref{CorrectResults subroutine}. We first prove that if $m_1$ and $m_2$ are both estimations of some bits of $\dfrac{s_0}{r}$ with $\dfrac{s_0}{r}=0.a_1a_2\cdots a_{\frac{L}{2}+1}$, then the output $m$ is perfect (i.e. $m=a_1 a_2\cdots a_{\frac{L}{2}}a_{\frac{L}{2}+1}0\cdots 0$), and the probability of this case is not less than $\dfrac{1}{r}$.

\begin{proposition}\label{s/r_integer}
Let $s_0\in \{0,1,\cdots, r-1\}$ satisfy that $2^{\frac{L}{2}+1}\cdot\dfrac{s_0}{r}$ is an integer, that is, $\dfrac{s_0}{r}=0.a_1a_2\cdots a_{\frac{L}{2}+1}$ where $a_i\in\{0,1\}$, $i=1,2,\cdots,\dfrac{L}{2}+1$. Then in {\rm Algorithm \ref{Distributed_order_finding}}, it holds that
\begin{equation}
{\rm Prob}(m=a_1 a_2\cdots a_{\frac{L}{2}+1}0\cdots 0)\geq \dfrac{1}{r}.
\end{equation}
\end{proposition}
\begin{proof}
Since the fractional part of $\dfrac{s_0}{r}$ is at most $(\dfrac{L}{2}+1)$-bit, in {\rm Algorithm \ref{Distributed_order_finding}}, we have
\begin{equation}
|\psi_{t_1,s_0/r}\rangle=|a_1 a_2\cdots a_{\frac{L}{2}+1}0\cdots 0\rangle
\end{equation}
 and
\begin{equation}
|\psi_{t_2,2^{\frac{L}{2}-1}s_0/r}\rangle=|a_{\frac{L}{2}}a_{\frac{L}{2}+1}0\cdots 0\rangle.
\end{equation}
Let $x=a_1 a_2\cdots a_{\frac{L}{2}}a_{\frac{L}{2}+1}0\cdots 0$ and $y=a_{\frac{L}{2}}a_{\frac{L}{2}+1}0\cdots 0$. By Equation (\ref{measure_orthonormal}), we have
\begin{align}
{\rm Prob}(m_1=x\ \text{and}\ m_2=y)&=\|| x\rangle \langle x| \otimes |y\rangle \langle y|\otimes I\ |\phi_{final}\rangle\|^2\\
&\geq\|| x\rangle \langle x| \otimes |y\rangle \langle y|\otimes I \ \dfrac{1}{\sqrt{r}}|\psi_{t_1,s_0/r}\rangle|\psi_{t_2,2^{\frac{L}{2}-1}s_0/r}\rangle|u_{s_0}\rangle\|^2\\
&= \dfrac{1}{r}.
\end{align}
Since $CorrectResults(x,y)=a_1 a_2\cdots a_{\frac{L}{2}}a_{\frac{L}{2}+1}0\cdots 0$, the lemma holds.
\end{proof}

	Then we prove that if $m_2$ is an estimation of the $\dfrac{L}{2}$th to $(2L+1)$th bit of $\dfrac{s_0}{r}$, we can get $IBits(m_2,1,2)=FBits(\dfrac{s_0}{r},\dfrac{L}{2},\dfrac{L}{2}+1)$.

\begin{lemma}\label{m_2}
Let $s_0\in\{0,1,\cdots, r-1 \}$ satisfy that $2^{\frac{L}{2}+1}\cdot\dfrac{s_0}{r}$ is not an integer and let $m_2$ satisfy
\begin{equation}\label{eq:m2}
d_{t_2}(m_2,FBits(\dfrac{s_0}{r},\dfrac{L}{2},2L+1+p))<2^{p}.
\end{equation}
Then $IBits(m_2,1,2)=FBits(\dfrac{s_0}{r},\dfrac{L}{2},\dfrac{L}{2}+1)$.
\end{lemma}
\begin{proof}
Since $2^{\frac{L}{2}+1}\cdot\dfrac{s_0}{r}$ is not an integer, we have
\begin{equation}\label{eq:s/r eatimate}
2^{-L}<\dfrac{1}{r}\leq \dfrac{2^{\frac{L}{2}+1}s_0 \bmod r}{r}\leq \dfrac{r-1}{r}<1-2^{-L}.
\end{equation}
 So we get $FBits(\dfrac{s_0}{r},\dfrac{L}{2}+2,\dfrac{3L}{2}+1)$ is not $00\cdots 0$ or $11 \cdots 1$. Hence, $FBits(\dfrac{s_0}{r},\dfrac{L}{2}+2, 2L+1)$ is not $00\cdots 0$ or $11 \cdots 1$. That is to say, if we add or subtract 1 to $FBits(\dfrac{s_0}{r},\dfrac{L}{2},2L+1)$, its first two bits are not changed. Thus by Equation (\ref{eq:m2}), we have
\begin{equation}
IBits(m_2,1,2)=FBits(\dfrac{s_0}{r},\dfrac{L}{2},\dfrac{L}{2}+1).
\end{equation}
Therefore the lemma holds.
\end{proof}

	If $IBits(m_2,1,2)=FBits(\dfrac{s_0}{r},\dfrac{L}{2},\dfrac{L}{2}+1)$, that is, the first two bits of $m_2$ are correct, then we can use these two bits of $m_2$ to ``correct" $m_1$. The following lemma can be used to show the correctness of Algorithm \ref{CorrectResults subroutine}.

\begin{lemma}\label{CorrectResults_correct}
Let $t>2$ be a positive integer and let $x,y$ be two $t$-bit strings with $d_{t}(x,y)\leq 1$. Then there only exists one element $b_0$ in $\{-1,0,1\}$ such that $(x+b_0) \bmod 2^{t}=y$, and for any $b\in\{-1,0,1\}$, $(x+b) \bmod 2^{t}=y$ if and only if $(IBits(x,t-1,t)+b) \bmod 2^2=IBits(y,t-1,t)$.
\end{lemma}
\begin{proof} By Lemma \ref{d_t}, we know that there exists such a $b_0$. It is clear that such a $b_0$ is unique. Next we prove that for any $b\in\{-1,0,1\}$, $(x+b) \bmod 2^{t}=y$ if and only if $(IBits(x,t-1,t)+b) \bmod 2^2=IBits(y,t-1,t)$. For any $b\in\{-1,0,1\}$, suppose $(x+b) \bmod 2^{t}=y$, then we have
\begin{equation}
(x+b) \bmod 2^{2}=y \bmod 2^{2}.
\end{equation}
That is,
\begin{equation}
(IBits(x,t-1,t)+b) \bmod 2^2=IBits(y,t-1,t).
\end{equation}
On the other hand, for any $b\in\{-1,0,1\}$, suppose $(IBits(x,t-1,t)+b) \bmod 2^2=IBits(y,t-1,t)$. Since there only exists one elements $b_1$ in $\{-1,0,1\}$ such that $(IBits(x,t-1,t)+b_1) \bmod 2^2=IBits(y,t-1,t)$,  $b$ is equal to $b_0$, that is, $b$ satisfies $(x+b) \bmod 2^{t}=y$. Consequently, the lemma holds.
\end{proof}

We can inspect Lemma \ref{CorrectResults_correct} from another aspect. If $d_{\frac{L}{2}+1}(m_1,FBits(\dfrac{s_0}{r},1,\dfrac{L}{2}+1))\leq 1$ and $IBits(m_2,1,2)=FBits(\dfrac{s_0}{r},\dfrac{L}{2},\dfrac{L}{2}+1)$ hold for some $s_0$, then the $CorrectionBit$ in Algorithm \ref{CorrectResults subroutine} exists, and $m_{prefix}=FBits(\dfrac{s_0}{r},1,\dfrac{L}{2}+1)$ holds as well.

\begin{proposition}\label{s/r_not_integer}
Let $m_2$ satisfy $d_{t_2}(m_2,FBits(\dfrac{s_0}{r},\dfrac{L}{2},2L+1+p))<2^{p}$ for some $s_0\in\{0,1,\cdots,r-1\}$ with $2^{\frac{L}{2}+1}\cdot\dfrac{s_0}{r}$ being not an integer. Suppose $d_{t_1}(m_1,FBits(\dfrac{s_0}{r},1,t_1))< 2^{p}$. Then $|\dfrac{m}{2^{2L+1+p}}-\dfrac{s_0}{r}|\leq 2^{-(2L+1)}$.
\end{proposition}
\begin{proof}
Since $d_{t_2}(m_2,FBits(\dfrac{s_0}{r},\dfrac{L}{2},2L+1+p))<2^{p}$ and $2^{\frac{L}{2}+1}\cdot\dfrac{s_0}{r}$ is not an integer, by Lemma $\ref{m_2}$, we have
\begin{equation}\label{m2 first two same}
IBits(m_2,1,2)=FBits(\dfrac{s_0}{r},\dfrac{L}{2},\dfrac{L}{2}+1).
\end{equation}
Since $d_{t_1}(m_1,FBits(\dfrac{s_0}{r},1,t_1))< 2^{p}$, by Lemma \ref{d_t}, we have
\begin{equation}
d_{\frac{L}{2}+1}(IBits(m_1,1,\dfrac{L}{2}+1),FBits(\dfrac{s_0}{r},1,\dfrac{L}{2}+1))\leq 1.
\end{equation}
As a result, in Algorithm \ref{CorrectResults subroutine}, the $CorrectionBit$ exists. By Equation (\ref{m2 first two same}), Lemma \ref{CorrectResults_correct}, and the steps 1 to 2 in Algorithm \ref{CorrectResults subroutine}, we get
\begin{equation}\label{eq:mprefix}
IBits(m,1,\dfrac{L}{2}+1)=m_{prefix}=FBits(\dfrac{s_0}{r},1,\dfrac{L}{2}+1).
\end{equation}
Since $m=m_{prefix}\circ IBits(m_2,3,t_2)$, by Equation (\ref{m2 first two same}) and Equation (\ref{eq:mprefix}), we get
\begin{equation}\label{eq:m}
d_{2L+1+p}(m,FBits(\dfrac{s_0}{r},1,2L+1+p))=d_{\frac{3L}{2}+2+p}(m_2,FBits(\dfrac{s_0}{r},\dfrac{L}{2},2L+1+p)< 2^{p}.
\end{equation}
Since $\dfrac{s_0}{r}$ is not an integer, similar to Equation (\ref{eq:s/r eatimate}), we know that $FBits(\dfrac{s_0}{r}, 1,2L+1)$ is not $00\cdots 0$ or $11\cdots 1$. Then by Equation (\ref{eq:m}), we get $d_{2L+1+p}(m,FBits(\dfrac{s_0}{r},1,2L+1+p))=|m-FBits(\dfrac{s_0}{r},1,2L+1+p)|$. Therefore, by Equation (\ref{eq:m}) and Lemma $\ref{result_acuracy}$, we obtain
\begin{equation}
|\dfrac{m}{2^{2L+1+p}}-\dfrac{s_0}{r}|\leq 2^{-(2L+1)}.
\end{equation}
\end{proof}

\begin{theorem}\label{correctness_for_distributed_order_finding}
In {\rm Algorithm \ref{Distributed_order_finding}}, for any fixed $s_0\in \{0,1,\cdots,r-1\}$, the probability of
 $|\dfrac{m}{2^{2L+1+p}}-\dfrac{s_0}{r}|\leq 2^{-(2L+1)}$  is at least $\dfrac{1-\epsilon}{r}$. The probability that there exists an $s\in \{0,1,\cdots,r-1\}$ such that $|\dfrac{m}{2^{2L+1+p}}-\dfrac{s_0}{r}|\leq 2^{-(2L+1)}$ is at least $1-\epsilon$.
\end{theorem}
\begin{proof}
By Proposition \ref{s/r_integer},  for any fixed $s_0\in \{0,1,\cdots,r-1\}$ with $2^{\frac{L}{2}+1}\cdot \dfrac{s_0}{r}$ being an integer, we have
\begin{equation}
{\rm Prob}((\dfrac{m}{2^{2L+1+p}}=\dfrac{s_0}{r})\geq\dfrac{1}{r}.
\end{equation}
For any fixed $s_0\in \{0,1,\cdots,r-1\}$ with $2^{\frac{L}{2}+1}\cdot \dfrac{s_0}{r}$ being not an integer, by Proposition \ref{phase_estimation_result} and Equation (\ref{measure_orthonormal}), we get that the probabilty of $d_{t_2}(m_2,FBits(\dfrac{s_0}{r},\dfrac{L}{2},2L+1+p))<2^{p}$ and $d_{t_1}(m_1,FBits(\dfrac{s_0}{r},1,t_1))< 2^{p}$ is at least $\dfrac{1}{r}(1-\epsilon')^2=\dfrac{1}{r}(1-\dfrac{\epsilon}{2})^2>\dfrac{1-\epsilon}{r}$. 
Consequently,  by Proposition \ref{s/r_not_integer}, we obtain
\begin{equation}
{\rm Prob}(|\dfrac{m}{2^{2L+1+p}}-\dfrac{s_0}{r}|\leq 2^{-(2L+1)})> \dfrac{1-\epsilon}{r}.
\end{equation}
Finally, the theorem has been proved.
\end{proof}

%By comparing Theorem \ref{correctness_for_distributed_order_finding} and Proposition \ref{order_finding_result}, we know that our distributed order-finding algorithm has similar effect as order-finding algorithm. Therefore order-finding algorithm can be replaced by our distributed order finding algorithm.

\section{Complexity analysis} \label{sec:complexity analysis}

   The complexity of the circuit of (distributed) order-finding algorithm depends on the construction of $C_t(M_a)$. There are two kinds of implementation of $C_t(M_a)$ proposed by Shor \cite{shor1999polynomial}. The first method (denoted as method (I)) needs time complexity $O(L^3) $ and space complexity $O(L)$, and the second method (denoted as method (II) ) needs time complexity $O(L^2\log L \log\log L) $ and space complexity $O(L\log L \log\log L)$. In this section, we compare our distributed order-finding algorithm with the traditional order-finding algorithm. For a more concrete comparison, we consider that $C_t(M_a)$ is implemented by method (I). There is a concrete implementation of order-finding algorithm by using method (I) in \cite{yimsiriwattana2004distributed}. However, the advantages of our distributed order-finding algorithm in space and circuit depth are independent of whether method (I) or method (II) is used.

\textbf{Space complexity} The implementation of the operator $C_t(M_a)$ in method (I) needs $t+L$ qubits plus $b$ auxiliary qubits for any positive integer $a$, where $b$ is $O(L)$. By Theorem \ref{quantum_telepotation}, to teleport $L$ qubits, computers $A$ and $B$ need to share $L$ pairs of EPR states and communicate with  $2L$ classical bits. As a result, $A$ needs $\dfrac{5L}{2}+1+\lceil\log_2(2+\dfrac{1}{\epsilon}\rceil)+b$ qubits and $B$ needs $\dfrac{5L}{2}+2+\lceil\log_2(2+\dfrac{1}{\epsilon})\rceil+b$ qubits. As a comparison, order-finding algorithm needs $3L+1+\lceil\log_2(2+\dfrac{1}{2\epsilon})\rceil+b$ qubits. So, our distributed order-finding algorithm can reduce nearly $\dfrac{L}{2}$ qubits.

\textbf{Time complexity.} The operator $C_t(M_a)$ can be implemented by means of $O(tL^2)$ elementary gates in method (I). Hence the gate complexity (or time complexity) in both our distributed order-finding algorithm and order-finding algorithm is $O(L^3)$.

\textbf{Circuit depth}. By Figure \ref{fig:controlled-U}, we know that the circuit depth of $C_t(M_a)$  depends on the circuit depth of  controlled-$M_a^{2^x} (x=0,1,\cdots,t-1)$ and $t$. The circuit depth of controlled-$M_a^{2^x}$ is $O(L^2)$ in method (I). By  observing the value ``$t$" in order-finding algorithm and our distributed order-finding algorithm, we clearly get that the circuit depth of each computer in our distributed order-finding algorithm is less than the traditional order-finding algorithm, even though both are $O(L^3)$.

\textbf{Communication complexity.} In our distributed Shor's algorithm, we need to teleport $L$ qubits. Therefore, the communication complexity of our distributed Shor's algorithm is $O(L)$. As a comparison, the communication complexity of the distributed order-finding algorithm proposed in \cite{yimsiriwattana2004distributed} is $O(L^2)$.

\section{Conclusions} \label{sec:conclusions}
	In this paper, we have proposed a new distributed Shor's algorithm. More specifically, we have proposed a new distributed order-finding algorithm. In this distributed quantum algorithm, two computers work sequentially via quantum teleportation. Each of them can obtain an estimation of partial bits of $\dfrac{s}{r}$  for some $s\in\{0,1,\cdots,r-1\}$ with high probability. It is worth mentioning that they can also be executed in parallel to some extent. We have shown that our distributed algorithm has advantages over the traditional order-finding algorithm in space and circuit depth. Our distributed order-finding algorithm can reduce nearly $\dfrac{L}{2}$ qubits and reduce the circuit depth to some extent for each computer. However, unlike parallel execution,  the way of serial execution that has been used in our algorithm leads to noise in both computers. 

We have proved the correctness of this distributed algorithm on two computers, a natural problem is whether or not this method can be generalized to multiple computers or to other quantum algorithms. We would further consider  the problem in subsequent study.

\section{Acknowledgements}
 This work  is  partly supported by the National Natural Science Foundation of China (Nos. 61572532,  61876195) and the Natural Science
	Foundation of Guangdong Province of China (No. 2017B030311011).

\end{document}